\documentclass[12pt]{article}
\usepackage{amsmath, amsfonts, url, amsthm}
\usepackage[T1]{fontenc}
\usepackage[latin1]{inputenc}
\usepackage{amssymb}

\begin{document}


\newtheorem*{theorem*}{Theorem}
\newtheorem{theorem}{Theorem}
\newtheorem{lemma}[theorem]{Lemma}
\newtheorem{claim}[theorem]{Claim}
\newtheorem{corollary}[theorem]{Corollary}
\newtheorem{conjecture}[theorem]{Conjecture}
\newtheorem{question}[theorem]{Question}
\newtheorem{problem}[theorem]{Problem}
\newtheorem{proposition}[theorem]{Proposition}
\newtheorem{axiom}[theorem]{Axiom}
\newtheorem{remark}[theorem]{Remark}
\newtheorem{example}[theorem]{Example}
\newtheorem{fact}[theorem]{Fact}
\newtheorem{exercise}[theorem]{Exercise}
\newtheorem{definition}[theorem]{Definition}
\newtheorem{observation}[theorem]{Observation}




\newcommand{\NChooseM}[2]{\ensuremath{\lp{(}\begin{array}{cc}#1\\#2\end{array}\rp{)}}}
\newcommand{\N}{\ensuremath{\mathbb{N}}}
\newcommand{\R}{\ensuremath{\mathbb{R}}}
\newcommand{\Z}{\ensuremath{\mathbb{Z}}}
\newcommand{\F}{\ensuremath{\mathbb{F}}}
\newcommand{\Ex}{\ensuremath{\mathbb{E}}}
\newcommand{\lp}[1]{\left #1}
\newcommand{\rp}[1]{\right #1}
\newcommand{\ZO}{\ensuremath{(0,1)}}
\newcommand{\pms}{\ensuremath{\{\pm 1\}}}

\def\mp{\{-1,+1\}}
\def\01{\{0,1\}}
\def\TF{\{T,F\}}

\newcommand{\braket}[2]{\langle #1, #2 \rangle}

\newcommand{\bra}[1]{\langle#1|}
\newcommand{\ket}[1]{|#1\rangle}

\newcommand{\Cor}{\mathrm{Corr}}

\newcommand{\rank}{\mathrm{rank}}
\newcommand{\rk}{\mathrm{rank}}
\newcommand{\nnrk}{\rk^+}
\newcommand{\trn}[1]{\ensuremath{\| #1 \|_{tr}}}
\newcommand{\fro}[1]{\ensuremath{\| #1 \|_{F}}}
\newcommand{\Tr}{\mathrm{Tr}}
\newcommand{\sn}[1]{\ensuremath{\| #1 \|}}
\newcommand{\recn}{\ensuremath{\mu}}
\newcommand{\nucn}{\ensuremath{\nu}}
\newcommand{\size}{\mathrm{size}}
\newcommand{\disc}{\mathrm{disc}}
\newcommand{\D}{D}

\newcommand{\DISJ}{\mathrm{DISJ}}
\newcommand{\UDISJ}{\mathrm{UDISJ}}
\newcommand{\OR}{\mathrm{OR}}

\newcommand{\basicsets}{\mathcal{S}}
\newcommand{\basicset}{S}

\newcommand{\norm}{\Phi}

\newcommand{\comnorm}{\mu_{\basicsets}}

\newcommand{\aenorm}{\xi}

\newcommand{\floor}[1]{\lfloor #1 \rfloor}
\newcommand{\ceil}[1]{\left\lceil #1 \right\rceil}

\newcommand{\ACC}{\mathrm{ACC}}
\newcommand{\AC}{\mathrm{AC}}

\newcommand{\NOF}{number-on-the-forehead }

\newcommand{\NIH}{number-in-hand }

\newcommand{\opt}{\mathrm{opt}}

\newcommand{\empha}[1]{{\em #1}}

\newcommand{\detconv}[1]{\mathcal{B}(D,#1)}

\newcommand{\CC}{\mathrm{CC}}
\newcommand{\G}{\mathcal{G}}
\newcommand{\gcm}{\mathcal{G}}

\newcommand{\ignore}[1]{}

\title{Nondeterministic Communication Complexity with Help and Graph Functions}

\author{Adi Shraibman \\ The School of Computer Science \\  
The Academic College of Tel Aviv-Yaffo\\{\tt adish@mta.ac.il}}

\date{}

\maketitle

\newcommand{\dom}{Y}
\newcommand{\color}{y}
\newcommand{\domdim}{m}
\newcommand{\indexf}{\ensuremath{\sigma}}
\newcommand{\simple}{\large^{-1}}

\newcommand{\helps}{h}
\newcommand{\helpsF}{\mathcal{H}}

\renewcommand{\dh}{D^{h}}
\newcommand{\dhp}{D^{h,*}}
\newcommand{\nh}{N^{h}}
\newcommand{\nhp}{N^{h,*}}

\abstract{We define nondeterministic communication complexity in the model of communication complexity
with help of Babai, Hayes and Kimmel \cite{BHK01}. We use it to
prove logarithmic lower bounds on the NOF communication complexity of explicit graph functions,
which are complementary to the bounds proved by Beame, David, Pitassi and Woelfel \cite{BDPW07}.}

\section{Introduction}

The Number On the Forehead model (NOF) in communication complexity
presents some of the more interesting and more challenging open questions
in communication complexity. In this model $k \ge 2$ players are each given
an input $x_i \in X_i$ ($i=1,\ldots,k$), and they require to compute a function
$f: X_1 \times X_2 \times \cdots \times X_k \to \{0,1\}$. Initially every player sees
all the inputs except their own. The players then communicate by taking turns in
writing one bit (0 or 1) on a blackboard. The communication ends when all the players
know the value of $f(x_1,\ldots,x_k)$. The {\em cost} of a protocol is the maximal number
of bits the players write on the blackboard during the computation of $f(x_1,\ldots,x_k)$, over
all choices of inputs $(x_1,\ldots,x_k)$. The {\em deterministic communication complexity}
of $f$, denoted $D_k(f)$, is equal to the minimal cost of a protocol for $f$.

We also denote by 
$N^1_k(f)$ the nondeterministic communication
complexity of $f$ in the $k$-players NOF model.
Nondeterministic protocols are more powerful than
deterministic ones. In addition to the input that is distributed between
the players in the NOF fashion, in nondeterministic communication 
complexity the players also have access to a number of bits given by an all powerful prover.
On input $(x_1,\ldots,x_k)$ such that
$A(x_1,\ldots,x_k)=1$ a correct protocol is required to have
at least one nondeterministic choice (proof) for which the output 
of the protocol is $1$.
If $A(x_1,\ldots,x_k)=0$ then all nondeterministic choices must
lead to the output $0$.

In randomized communication complexity the players are allowed to use random bits.
The inputs are distributed as in the deterministic model, and the players communicate the same way by writing on a blackboard.
The next bit of each player is dependent on the part of the input that he sees, previous
communication, and the random bits. At the end of the communication the players deduce
the output from the communication transcript. Note that the output is now a random variable. 
It is required that the players deduce the correct
value of $f(x_1,\ldots,x_k)$ with probability at least $2/3$ for every input
$(x_1,\ldots,x_k)$. 

A fundamental problem in multiparty communication complexity, as in many computational models, is to study the power
of randomization. Beame, David, Pitassi and Woelfel \cite{BDPW07} showed a non-constructive separation between
randomized and nondeterministic NOF communication complexity. In fact they showed this gap in a very simple
family of functions they called {\em graph functions}.
A function $f:[n]^{k-1} \times [N] \to \{0,1\}$ \footnote{We assume here that the input space is $[n]^{k-1} \times [N]$. This is just for simplicity of 
presentation, and all the definitions and results hold for a general input space.}
is a graph function if for every $(x_1,\ldots,x_{k-1})$ there is a unique $y \in N$ such that
$f(x_1,\ldots,x_{k-1}, y) = 1$.

An advantage of graph functions, observed in \cite{BDPW07} is that the randomized communication complexity 
of any graph function is $O(1)$.
Thus, in order to separate randomized from nondeterministic communication 
complexity it is enough to prove a large lower bound on the nondeterministic communication complexity 
of any graph function. Beame et al \cite{BDPW07} used an elegant counting
argument to prove that most graph functions $f: [n]^{k-1} \times [N] \to \{0,1\}$ with $N \cong \sqrt{\frac{n}{k}}$
have nondeterministic communication complexity $\Omega(\log \frac{n}{k})$. 
It remains a challenging problem though to 
present an explicit function exhibiting a large gap, even for $k=3$. 

Another nice aspect of graph functions is that they can be alternatively viewed as a $(k-1)$-dimensional object.
A graph function $f: [n]^{k-1} \times [N] \to \{0,1\}$ is associated with the function $A: [n]^{k-1} \to [N]$
defined by $A(x_1,\ldots,x_{k-1}) = y$ for the unique $y \in [N]$ satisfying $f(x_1,\ldots,x_{k-1},y)=1$.
We let $A = Base(f)$ denote this base function, and also write $f = Lift(A)$.
\footnote{In \cite{BDPW07} a different notation is used, they write $g$ instead of $A$
and $f=graph^g$. We use the notation $f=Lift(A)$ since we consider other lift options in Section~\ref{sec:other-lifts}.}

It is particularly hard to prove lower bounds for high-dimensional permutations 
and linjections \cite{hdp17} which are a special
type of graph functions. For these functions $N \ge n$ while the results of 
\cite{BDPW07} as well as ours apply only when $N \ll n$.
For illustration, two-dimensional permutations are the class of functions $f:[n]^3 \to \{0,1\}$
for which $Base(f)$ is a Latin square. Improving the known lower bounds for permutations
(even two-dimensional) imply strong applications even beyond the scope of communication complexity. 

The highest lower bound for the communication complexity of an explicit graph function
$f : [n]^{k-1} \times [N] \to \{0,1\}$ is $\Omega(\log \log n)$ proved in \cite{BDPW07}.
For permutations the best lower bound is $\Omega(\log \log \log n)$
for $k=3$ proved in \cite{hdp17}, 
and also in \cite{beigel2006multiparty} for Exact-T functions which are a
special type of permutations. These bounds are also closely related to 
the results of \cite{graham2006monochromatic} and to Proposition~4.3 in \cite{alon2012nearly}.
For $k > 3$ the best lower bound for the communication 
complexity of any permutation is $\Omega(log^* n)$ \cite{hdp17}.

The aim of this manuscript is to present another approach 
for proving lower bounds on the deterministic and nondeterministic communication complexity of graph functions.
In particular we give an alternative proof to the 
$\Omega(\log \log n)$ bound of \cite{BDPW07}. The bounds of 
\cite{BDPW07}, both constructive and non-constructive, use an observation
that a nondeterministic protocol for a graph function can always be put in a special
{\em normal form}. Namely, a graph function always has a very simple type of protocol
in which one of the players is oblivious
and the others send only one bit. This protocol was also previously used
for a specific graph function by Chandra, Furst and Lipton in \cite{CFL83}
and for a more general family of permutations in \cite{beigel2006multiparty}.

One way to tackle the problem of proving lower bounds on the deterministic and nondeterministic communication complexity of 
explicit graph functions,  is to consider a relaxation of the model.
The above mentioned one-way protocol for graph functions suggests to use the model of 
{\em communication complexity with help} defined by Babai, Hayes and Kimmel \cite{BHK01}.
In this model $k$ players wish to evaluate a function $A : [n]^{k} \to [N]$. 
A deterministic communication protocol with help is similar to the NOF protocol
described earlier, with the addition of a ``helper''.   
Before the players start
the communication on inputs $(x_1,\ldots,x_k)$, the helper sends them a help string
of at most $b$ bits, which can depend on any part of the input.
The cost of a protocol is the maximal, over all inputs, of the length
of the communication transcript, plus the length of the help string. The deterministic communication complexity with
$b$ help bits, denoted by $\dh_{k,b}(A)$, is the minimal cost of such a protocol for $A$.

Note that communication complexity with help is different than nondeterministic communication complexity
in that the players in this model do not need to verify the information given by the helper, this is simply free information.
Obviously, the helper can simply announce the value of the function with $\log N$ bits of information. Thus
the interesting question is how much communication is needed when the helper gives less than $\log N$ bits. 

Babai et al \cite{BHK01} used communication complexity with help in order to prove lower bounds 
on the one-way communication complexity of explicit functions. To prove the lower bounds 
they have defined a concept of multicolor discrepancy and used it as a lower bound
for $\dh_{k,b}(A)$. They then computed the multicolor discrepancy of certain functions, thus providing
lower bounds for the deterministic communication complexity with help of these functions.
We exploit these bounds and translate them also to lower bounds on
the deterministic complexity of explicit graph functions.

Let $f: [n]^{k-1} \times [N] \to \{0,1\}$ be a graph function, and let $A=Base(f)$. 
It is easy to check that $D_k(f) \le \dh_{k-1, b}(A) + 1$, for any natural number $b$. This upper bound
holds even for the one-way model, where the last player sends a message and
then the players communicate as usual but without the participation of the last player. 
Indeed the one-way NOF communication complexity model with $k$ players is stronger than the model with help
and $k-1$ players, since on input $(x_1,\ldots,x_k)$ the first $(k-1)$ players see the value on the forehead
of the $k$-th player, 
which is essentailly $A(x_1,\ldots,x_{k-1})$, and need only validate it.  
In communication complexity with help on the other hand, the $k$-th player is replaced by the helper, and the rest of
the players are required to compute $A(x_1,\ldots,x_{k-1})$. 
 
Our first result is that for graph functions, the gap between these two models cannot be arbitrary though, 
which makes this relaxation useful. 
\begin{theorem*}
\label{th:gf_dcch}
Let $f: [n]^{k-1} \times [N] \to \{0,1\}$ be a graph function and let $A = Base(f)$. Then
$$
D_k(f) \ge \min \{\dh_{k-1, b}(A) - (k-1)N, b\}.
$$
\end{theorem*}
The above lower bound, combined with the mentioned results of \cite{BHK01}, can give a lower bound of
$\Omega(\log \log n)$ on the deterministic communication complexity of explicit
graph functions, matching the bound of \cite{BDPW07}.
\begin{theorem*}
There is an explicit graph function 
$f : [n]^{k-1} \times [N] \to \{0,1\}$ such that 
\[
D_k(f) \ge \log N \ge \Omega(\log \log n - k).
\]
\end{theorem*}
Even though the above bound is tight, it is only applicable when 
$N$ is at most $c\log n$ for some constant $c < 1$. The reason for this limitation 
is that the protocol that gives the lower bound iterates
over all values $y \in [N]$ in order to find the unique value for which
$f(x_1,\ldots,x_{k-1},y)=1$. A natural approach to break this barrier is to define 
and use nondeterministic communication complexity with help, which we do in Section~\ref{sec:nondet_help}. The bound
$N^1_k(f) \le \nh_{k-1,b}(A)+1$ still naturally holds, and on the other hand we prove the following lower bound.
\begin{theorem*}
For every graph function $f: [n]^{k-1} \times [N] \to \{0,1\}$ it
holds that
\[
N^1_k(f) \ge \min \{\nh_{k-1,b}(A) -\log N - k +1, b\},
\]
where $A=Base(f)$.
\end{theorem*}

Nondeterministic communication complexity with help captures better the communication complexity
of graph functions, and provides a lower bound that allows a much wider range for $N$. 
In fact, a tight lower bound for $\nh_{k-1,b}(A)$ can provide a tight lower bound
for the communication complexity of the corresponding graph function $f: [n]^{k-1} \times [N] \to \{0,1\}$,
as long as $N \ll n$. Recall that currently only exponentially smaller lower bounds are known.
This makes proving lower bounds for this model an interesting question.
The first place to look for such lower bounds, is to rely on the known bounds for the deterministic model.
In the classical two players boolean model it is known that 
$D_2(f) \le O(N^1_2(f)N^1_2(\bar{f}))$ for every function $f: [n]^2 \to \{0,1\}$.
But the proof breaks down for communication complexity with help. Even for regular
protocols, it is not clear whether this bound can be generalized to $k \ge 3$ players in the NOF model.
It is also an interesting and nontrivial question whether multicolor discrepancy provides 
a good lower bound for nondeterministic communication complexity with help. The 
much weaker bound on nondeterministic communication complexity via deterministic 
complexity \cite[Ex. 2.6]{KN97} though, 
can be adapted also to the case of communication complexity with help.
\begin{theorem*}
\label{111}
Let $A: [n]^k \to [N]$ be a function and let $b < \log N$ be a natural number. Let 
$\dh_{k,b}(A) = b + c_d$ where $b$ is the number of help bits
and $c_d$ is the number of subsequent bits of communication, in an optimal communication protocol.
Similarly let $\nh_{k,b}(A) = b + c_n$, then
\footnote{According to the definition, the number of help bits can be smaller than $b$, we thus
need to justify why there is always an optimal protocol with exactly $b$ help bits. We remark on that at the end
of Section~\ref{sec_partial_functions}.} 
\[
c_d \le (k-1)2^{c_n} + c_n.
\] 
\end{theorem*}
This yields:
\begin{theorem*}
Let $f: [n]^{k-1} \times [N] \to \{0,1\}$ be a graph function, let $A = Base(f)$, and let $b = \log N - 1$. 
Then
\[
N^1_k(f) \ge \min \left\{ \log \left(\dh_{k-1,b}(A) - \log N \right) - \log k - k, \log N \right\}.
\]
\end{theorem*}
Together with the discrepancy lower bound of \cite{BHK01}, the above inequality implies that
$N^1_k(f) \ge \Omega(\log \log n)$ for any graph function $f: [n]^{k-1} \times [N] \to \{0,1\}$ with a base function 
that has multicolor discrepancy smaller than $\frac{1}{N^{1+\Omega(1)}}$, similarly to the bounds of \cite{BDPW07}. 
The techniques of \cite{BDPW07} are different though and the underlying statement is complementary
to ours. They prove that an optimal protocol requires $\Omega(\log \log n)$ help bits when the discrepancy is smaller than 
$\frac{1}{N^{1+\Omega(1)}}$. Our bound on the other hand says that regardless of the number of help bits,
even if $\log N -1$ help bits are given, the subsequent communication between the players
has complexity $\Omega(\log \log n)$. 
We describe previous results in more detail in Section~\ref{sec:prev_res}.

Finally, in Section~\ref{sec:other-lifts} we briefly consider alternative Lift options for 
$A: [n]^{k} \to [N]$, other than the corresponding graph function.

\section{Graph functions and deterministic communication complexity with help}

We first prove the following lower bound and then apply it to give lower bounds
on explicit graph functions.
\begin{theorem}
\label{th:bound_help_reg_det}
Let $f: [n]^{k-1} \times [N] \to \{0,1\}$ be a graph function and let $A = Base(f)$. Then
$$
D_k(f) \ge \min \{\dh_{k-1, b}(A) - (k-1)N, b\}.
$$
\end{theorem}

\begin{proof}
If $D_k(f) > b$ then we are done. Otherwise assume
that $D_k(f) \le b$, we show that in this case
\[
\dh_{k-1, b}(A) \le D_k(f) + (k-1)N.
\]
To prove this lower bound,
let $P$ be an optimal communication protocol for
$f$. We define the following protocol for $A$:
On input $(x_1,\ldots,x_{k-1})$ the players iterate over all values $\color \in [N]$,
and check whether $\color$ is equal to $A(x_1,\ldots,x_{k-1})$. In each iteration, 
the players use as help-bits the transcript $\mathcal{T}$ of
the run of the protocol $P$ on $(x_1,\ldots,x_{k-1},A(x_1,\ldots,x_{k-1}))$.
Each player compares his actions according to the transcript
$\mathcal{T}$ with his actions according to $P$ on $(x_1,\ldots,x_{k-1},\color)$,
and announces whether or not they agree. If for
some player these actions do not agree this must mean that $\color \ne
A(x_1,\ldots,x_k)$.

Otherwise, since the actions of the $k$-th player do not depend on the $k$-th input,
$\mathcal{T}$ is the transcript of a run of $P$ both on $(x_1,\ldots,x_{k-1},\color)$
as well as on the input $(x_1,\ldots,x_{k-1},A(x_1,\ldots,x_{k-1}))$.
Since the transcript of the protocol
determines its output, and $A$ is a graph function,
the players can find this way the unique value $y$ for which $\color = A(x_1,\ldots,x_k)$.

This protocol uses $N$
rounds of communication, one round for each value
$\color \in [N]$. In each round the protocol uses $k-1$ bits of communication.  In addition
the protocol uses $D_k(f)$ help
bits (Recall we have assumed that $D_k(f) \le b$).
Thus, $D^h_{k-1,b}(A) \le (k-1)N+D_k(f)$ as required.
\end{proof}
Babai et al \cite{BHK01} proved a lower bound on distributional communication
complexity with help in terms of multicolor discrepancy.
They also gave explicit functions with low discrepancy.

\paragraph{Multicolor discrepancy}
Let $A: X \to \dom$ be a function. For a subset $S \subset X$ and
an element $\color \in \dom$, define
\[
disc(A,S,\color) = \left| |A^{-1}(\color) \cap S| - |S|/|B|\right| /|X|.
\]
The discrepancy of a set $S$ is
\[
disc(A,S) = \max_{\color \in \dom} disc(A,S,y).
\]
The discrepancy of a set system $\mathcal{F}$ is defined as
\[
disc(A,\mathcal{F}) = \max_{S \in \mathcal{F}} disc(A,S).
\]
In words, discrepancy measures how much the size of
$A^{-1}(\color) \cap S$ deviates from what is
expected from a random function $A$.

We are interested in the case where $X = [n]^{k-1}$, $Y=[N]$
and $\mathcal{F}$ is the family of cylinder intersections. We denote the discrepancy
of a function $A: [n]^{k-1} \to [N]$ simply by $disc_{k-1}(A)$.
The following bound is proved in \cite{BHK01}.
\begin{theorem}[\cite{BHK01}]
\label{th:D_disc}
For every function $A: [n]^{k-1} \to [N]$ 
\footnote{The result of \cite{BHK01} holds for distributional communication complexity with help,
but we only need the deterministic model.}
$$D^h_{k-1,b}(A) \ge \log \left( \frac{1-(2^b/N)}{disc_{k-1}(A)} \right).$$
\end{theorem}

An example of an explicit function with small discrepancy is
\begin{definition}[\cite{BHK01}]
Let $q$ be a prime power, and let $d$ be a positive integer. Let $M_d$ be
the space of $d \times d$ matrices over $\F_q$. The function
$T_{q,d,k}: M_d^k \to \F_q$ is defined by
\[
T_{q,d,k}(B_1,\ldots,B_k) = Tr(B_1\cdot B_2 \cdot \ldots \cdot B_k).
\]
\end{definition}

\begin{lemma}[\cite{BHK01}]
\label{cor:TMP}
$-\log disc_k(T_{q,d,k}) \ge \Omega(\frac{d^2\log q}{k^22^k})$.
\end{lemma}
Combining these facts with Theorem~\ref{th:bound_help_reg_det} gives:
\begin{corollary}
Let $N$ be a prime power, $k$ be an integer, and take 
$d = c \cdot k^{3/2}2^{k/2} \cdot \sqrt{\frac{N}{\log N}}$ for
large enough $c$. Let $A = T_{N,d,k}$ and $f = Lift(A)$, then
\[
D_k(f) \ge \log N \ge \Omega(\log \log n - k),
\]
where the domain of $f$ is $[n]^{k-1} \times [N]$.
\end{corollary}

\begin{proof}
By Theorem~\ref{th:D_disc} and Lemma~\ref{cor:TMP}
$$
\dh_{k-1}(A) \ge c_1\frac{d^2\log N}{k^22^k}, 
$$
for some constant $c_1>0$. But  $c_1\frac{d^2\log N}{k^22^k}= c_1c^2 kN$, therefore
if we choose $c=c_1^{-1/2}$ then $c_1c^2 = 1$ and $\dh_{k-1}(A) \ge kN$.

By Theorem~\ref{th:bound_help_reg_det}
$$
D_k(f) \ge \min \{\log N, \dh_{k-1}(A) - (k-1)N\}.
$$
Thus
$$
D_k(f) \ge \min \{\log N, N\} = \log N.
$$

Finally notice that $n  = 2^{d^2 \log N} = 2^{c^2 k^3 2^k N}$, is the size of the first
$k-1$ players input space. Thus
$$
\log \log n = \log N + k + 3\log k +2\log c. 
$$
\end{proof}

\section{Nondeterministic communication complexity with help}
\label{sec:nondet_help}

The protocol in the proof of Theorem~\ref{th:bound_help_reg_det}
iterates over values $\color \in [N]$ in search of the correct value.
This iteration adds an additive factor to the complexity, that is linear in $N$.
It seems natural to consider nondeterministic complexity for such a search problem,
as there is a potential of getting exponentially better lower bounds, and also improving
the dependency on $N$. In this section we define nondeterministic communication
complexity with help and use it to prove lower bounds on the deterministic NOF
communication complexity of graph functions.

We define {\em nondeterministic communication complexity with help} of a function
$A : [n]^{k} \to [N]$ similarly to deterministic communication. The difference is that the communication after 
receiving the help bits is nondeterministic. Namely, on input $(x_1,\ldots,x_k)$, after receiving
the help string, the 
players also receive a proof from an all powerful prover, and are then required to compute
the value of $A(x_1,\ldots,x_k)$. The output of the computation can either be the correct value
or ``don't know''. It is required that for every input there is at least one choice of a nondeterministic 
string for which the protocol outputs the correct answer.
We denote by $\nh_{k,b}(A)$ the 
nondeterministic communication complexity with help of $A$ with $b$ help bits.
We also let $N_k(A) = \nh_{k,0}(A)$, be the nondeterministic communication complexity of $A$.

\subsection{Bounds}

\begin{theorem}
\label{th:bound_help_reg}
Let $f: [n]^{k-1} \times [N] \to \{0,1\}$ be a graph function and let $A = Base(f)$. Then
\[
N^1_k(f) \ge \min \{\nh_{k-1,b}(A) -\log N - k + 1, b\}.
\]
\end{theorem}

The proof is similar to the proof of Theorem~\ref{th:bound_help_reg_det},
excluding the fact that here the deterministic search is replaced with a nondeterministic
choice.
\begin{proof}
If $N^1_k(f) > b$ then the bound follows. Assume therefore that
$N^1_k(f) \le b$. We prove that in this case
\[
\nh_{k-1,b}(A) \le N^1_k(f) + \log N + k -1.
\]
The proof works by defining an efficient communication protocol for
$\nh_{k-1,b}(A)$ based on a protocol for $N^1_k(f)$.
Let $P$ be an optimal nondeterministic communication protocol for
$f$. We define the following protocol for $A$.
On inputs $(x_1,\ldots,x_{k-1})$ the players guess an output $\color \in [N]$, and
then verify whether this is really the output. To verify whether $\color$
is the output, the players use as help-bits a transcript $\mathcal{T}$ of
a run of the protocol $P$ on input $(x_1,\ldots,x_{k-1},A(x_1,\ldots,x_{k-1}))$,
with nondeterministic choices that achieve the correct answer.
Each player compares his actions according to the transcript
$\mathcal{T}$ with his actions according to $P$ on inputs $(x_1,\ldots,x_k,\color)$,
and announces whether or not they agree. If for
some player these actions do not agree this must mean that $\color \ne
A(x_1,\ldots,x_{k-1})$ and thus the protocol outputs ``don`t know''.

Otherwise, since the actions of the $k$-th player do not depend on the $k$-th input,
$\mathcal{T}$ is the transcript of a run of $P$ both on inputs $(x_1,\ldots,x_{k-1},\color)$
and on inputs $(x_1,\ldots,x_{k-1},A(x_1,\ldots,x_{k-1}))$.
By our choice of nondeterministic bits, and since the transcript of the protocol
determines its output, if the protocol accepts it must be that $\color = A(x_1,\ldots,x_{k-1})$.
Note that here we use the fact that
$P$ makes only one-sided mistakes.

Finally, notice that this protocol uses $\log N + (k-1)$
bits of communication, and $N^1_k(f)$ help
bits. We therefore conclude that $\nh_{k-1,b}(A) \le N^1_k(k) + \log N  + k - 1$.
\end{proof}

As in the deterministic and one-way models, it also holds that:
\begin{theorem}
\label{th:bound_help_reg_2}
Let $f: [n]^{k-1}\times [N] \to \{0,1\}$ be a graph function and let $A  =Base(f)$, then
\[
N_k^1(f) \le \nh_{k-1,b}(A)+1.
\]
\end{theorem}

\begin{proof}
Given an input $(x_1,\ldots,x_k)$,
the $k$-th player sees all inputs on the other player's foreheads,  $(x_1,\ldots,x_{k-1})$. 
Thus, the last player can compute the help string and send it to the other players.
The first $k-1$ players then use an optimal protocol for $\nh_{k-1,b}(A)$
to compute $A(x_1,\ldots,x_{k-1})$.
If the result of the protocol is ``don`t know'' then the first player outputs $0$.
After the first $k-1$ players compute $A(x_1,\ldots,x_{k-1})$ they compare it with $x_{k}$.
If these quantities are equal the first player outputs $1$, otherwise he outputs $0$. 
This protocol requires $\nh_{k-1,b}(A)+1$ bits of communication. Note that there is no restriction on $b$ here.
\end{proof}

\section{Lower bounds for explicit graph functions}

\newcommand{\inp}[1][x]{\mathbf{#1}}

We show in this section a lower bound on nondeterministic communication complexity with help
in terms of deterministic complexity with help. This bound is a natural extension of the exponential bound
known for the binary case of the two players traditional model \cite[Ex. 2.6]{KN97}. This bound enables to prove lower 
bounds on explicit graph functions for which the base function has relatively low multicolor discrepancy.

\subsection{Partial functions}
\label{sec_partial_functions}

An alternative way to view communication complexity 
with $b$ help bits is that we are allowed to partition the input space into at most $2^b$
parts and compute the complexity of the partial function confined to any of the parts, separately.
The communication complexity with help
is equal to the maximal complexity over these partial problems, plus the logarithm of the size of the partition.

For the formal definition we first recall the definition of the communication complexity
of a partial function. Let $A: [n]^k \to [N]$ be a function and $S \subset [n]^k$ a subset
of the inputs. 
The communication complexity of $A$ restricted to S, denoted $CC(A,S)$,
where $CC$ is any communication complexity model,
is defied similarly to $CC(A)$ with the exception that a protocol only needs
to be correct on inputs that belong to $S$. 

Now, let $A: [n]^k \to [N]$ be a function, and let $b$ be a natural number,  
the communication complexity $\dh_{k,b}(A)$ is equal to
$$
\min_{\mathcal{S}}\left( t + \max_{i=1,\ldots,2^t} D_k(A,S_i) \right),
$$
where the minimum is over all partitions
$\mathcal{S}$ of $[n]^k$ into $2^t$ subsets $\{S_1,S_2,\ldots,S_{2^t}\}$,
with $t \le b$.
The partition $\mathcal{S}$ is defined by the help bits, all inputs in a single
part $S_i$ share the same help string.

The nondeterministic communication complexity $\nh_{k,b}(A)$ is defined similarly as
$$
\min_{\mathcal{S}}  \left(t + \max_{i=1,\ldots,2^t} N_k(A,S_i) \right).
$$
The major difficulty in proving lower bounds on nondeterministic communication complexity with help is that the
rectangles can intersect also outside the subset $S_i$, where there is no restriction on the value 
of the entries. 

\paragraph{The number of help bits} We note that we can assume without loss of generality that the number of help bits
is exactly $b$. That is, the size of the partition is $2^b$. We exhibit that on $\nh_{k-1,b}(A)$, the proof for $\dh_{k-1,b}(A)$ is similar.
\begin{proof}
Let $P$ be an optimal communication protocol for $\nh_{k-1,b}(A)$. Let $P_H$ be the help player's protocol
and $P_C$ the subsequent communication protocol. Namely, on input $(x_1,\ldots,x_{k-1})$ first the help player sends
$P_H(x_1,\ldots,x_{k-1})$ to the players and then the transcript of their communication is given by $P_C(x_1,\ldots,x_{k-1})$.
Let $\nh_{k-1,b}(A) = h + c$ where $h$ is the maximal length of a help string, and $c$ is the maximal length of a transcript of 
$P_C$. Then, if $h < b$, we can change the protocols and add to the
help string the initial $b-h$ communication bits of the transcript given by $P_C$, since the helper knows everything.
Thus, we can assume without loss of generality that $h = b$.
\end{proof}

\subsection{Cylinder intersections}

A key definition in multiparty communication complexity is that of a cylinder intersection. We say that $C \subseteq X_1 \times \cdots \times X_k$
is a {\em cylinder in the $i$-th coordinate} if membership in $C$ does not depend on the $i$-th coordinate.
Namely, for every $x,x' \in X_i$ there holds
$(a_1,\ldots,a_{i-1},x,a_{i+1},\ldots,a_k)\in C$ iff $(a_1,\ldots,a_{i-1},x',a_{i+1},\ldots,a_k)\in C$.
A cylinder intersection is a set $C$ of the form $C = \cap_{i=1}^k C_i$
where $C_i$ is a cylinder in the $i$-th coordinate.

Following are some well known basic facts regarding the relation between cylinder intersections 
and communication complexity:
\begin{lemma}[\cite{KN97}]
\label{cylinder_intersection_intersection}
\label{cylinder_intersection_membership}
There holds
\begin{enumerate}

\item Let $C=\cap_{i=1}^k C_i$ be a cylinder intersection in $X_1 \times \cdots \times X_k$ and let
$\inp \in X_1 \times \cdots \times X_k$. Then $\inp \in C$
if and only if $\inp\in C_i $ for all $i \in [k]$.

\item The above fact gives a one round protocol to determine membership 
in a cylinder intersection $C = \cap_{i=1}^k C_i$. Given an input $\inp \in X_1 \times \cdots \times X_k$
player $i$ checks whether $\inp \in C_i$, and transmits $1$
if it is true and $0$ otherwise. It holds that $\inp \in C$,
if and only if all players transmitted $1$.

\item 
Let $A:[n]^k \to [N]$ be a function, and let $S \subset [n]^k$ be a subset of the entries.
An optimal protocol for $N_k(A,S)$ induces a cover of $[n]^k$ 
by at most $2^{N_k(A,S)}$ cylinder intersections that are monochromatic with respect
to $A$ on $S$. 

\end{enumerate}
\end{lemma}

\subsection{Determinism versus nondeterminism}

\begin{theorem}
\label{th:_d_n}
Let $A: [n]^k \to [N]$ be a function and let $b < \log N$ be a natural number. Let 
$\dh_{k,b}(A) = b + c_d$ where $b$ is the number of help bits
and $c_d$ is the number of subsequent bits of communication, in an optimal communication protocol.
Similarly let $\nh_{k,b}(A) = b + c_n$, then 
\[
c_d \le (k-1)2^{c_n} + c_n.
\] 
\end{theorem}

Theorem~\ref{th:_d_n} is a direct consequence of the following lemma.
\begin{lemma}
\label{lem:partial_d_n}
Let $A: [n]^k \to [N]$ be a function and let $S \subset [n]^k$, then 
\[
D_k(A,S) \le (k-1)2^{N_k(A,S)}+N_k(A,S).
\] 
\end{lemma}

\begin{proof}[Proof of Lemma~\ref{lem:partial_d_n}]
Let $\chi  = 2^{N_k(A,S)}$ and let $\{C^j\}_{j=1}^{\chi}$ be an optimal cover for $N_k(A,S)$
that exists by Lemma~\ref{cylinder_intersection_membership} (part 3).
That is, a cover of $[n]^2$ into $\chi$ cylinder intersections that are monochromatic 
on $S$ with respect to $A$.

On input $(x_1,\ldots,x_k) \in S$ the players then do the following:

\begin{enumerate}

\item For $i=1,\ldots,k$: Player $i$ computes the vector $V_i \in \{0,1\}^{\chi}$, whose $j$th
coordinate is equal to $1$ if and only if $(x_1,\ldots,x_k)$ belongs to $C^j_i$

\item The $i$-th player writes $V_i$ on the blackboard, for $i=1,\ldots,k-1$.

\item The $k$-th player publishes the index of a cylinder intersection $C^j$ that 
contains $(x_1,\ldots,x_k)$. 

\end{enumerate}

Since $\{C^j\}_{j=1}^{\chi}$ is a cover for $N_k(A,S)$, there exists a
cylinder intersection $C_j$ that contains $(x_1,\ldots,x_k)$.  
By Lemma~\ref{cylinder_intersection_membership} (part 1), 
$C_j$ contains $(x_1,\ldots,x_k)$ if and only if the $j$th
coordinate of $V_i$ is equal to $1$ for every $i=1,\ldots,k$.
Since $\{C^j\}_{j=1}^{\chi}$ is also monochromatic on $S$,
the above protocol is correct, and when it ends all players know
$A(x_1,\ldots,x_k)$. 

The first step requires no communication, the second step uses $(k-1)\chi$ bits, and in the last
step the $k$-th player writes $\log \chi$ bits on the board. The total number of bits in a communication
is $(k-1)\chi + \log \chi$. Since $\chi = 2^{N_k(A,S)}$ the claim follows.
\end{proof}

\begin{proof}[Proof of Theorem~\ref{th:_d_n}]

Let $H = 2^{b}$ and let $\mathcal{S} = \{S_1,S_2,\ldots,S_{H}\}$ be a partition
that achieves the optimal complexity for $\nh_{k,b}(A)$. Namely,
$$
c_n = \max_{i=1,\ldots,H} N_k(A,S_i).
$$
By Lemma~\ref{lem:partial_d_n}, for every $i=1,\ldots,H$ it holds that
\[
D_k(A,S_i) \le (k-1)2^{N_k(A,S_i)}+N_k(A,S_i) \le (k-1)2^{c_n}+c_n. 
\]
In particular
\[
c_d \le \max_{i=1,\ldots,H} D_k(A,S_i) \le (k-1)2^{c_n}+c_n.
\]
\end{proof}

\subsection{A weak lower bound via multicolor discrepancy}

We prove a weak lower bound using deterministic communication complexity with help,
the lower bound via multicolor discrepancy then follows from Lemma~\ref{th:D_disc}.
\begin{theorem}
Let $f: [n]^{k-1} \times [N] \to \{0,1\}$ be a graph function, let $A = Base(f)$, and let $b = \log N - 1$. 
Then
\[
N^1_k(f) \ge \min \left\{ \log \left(\dh_{k-1,b}(A) - \log N \right) - \log k - k, \log N \right\}.
\]
\end{theorem}

\begin{proof}
Let $\nh_{k-1,b}(A) = (\log N -1) + c_n$, where $c_n$ is the number of communication bits
after the help string is given, similarly let $\dh_{k-1,b}(A) = (\log N -1) + c_d$.
Recall from the remark at the bottom of Section~\ref{sec_partial_functions},
we can assume without loss of generality that the number of help bits 
is exactly $b = \log N -1$.

By Theorem~\ref{th:_d_n}, it holds that
\[
c_d \le (k-1)2^{c_n}+c_n \le k2^{c_n}.
\]
Thus,
\[
c_n \ge \log c_d - \log k \ge \log \left(\dh_{k-1,b}(A) - \log N \right) - \log k.
\]
Finally note that by Theorem~\ref{th:bound_help_reg},
\[
N^1_k(f) \ge \min \{ \nh_{k-1,b}(A) -(\log N -1) - k, \log N\} = \min \{c_n - k, \log N\}.
\]
\end{proof}

\section{Previous results, a closer look}
\label{sec:prev_res}

In this section we review previous results in more detail.
We start with a few more of the properties of graph functions, that we need in order
to describe the previous results.
A pleasant aspect of the study of graph functions is that the communication complexity
is completely characterized by {\em stars}. For simplicity we describe this notion for the
case $k=3$. A star is a triplet $(x,y,z'), (x',y,z), (x,y',z)$ of points in $[n]\times [n] \times[N]$
such that $x \ne x'$, $y \ne y'$ and $z \ne z'$.
In the $2$-dimensional case stars become what is called an {\em $A$-star} \cite{hdp17}. 
The star $(x,y,z'), (x',y,z), (x,y',z)$ correspond to the $A$-star 
$(x,y), (x',y), (x,y')$, which is a triplet of distinct points such that $A(x',y) = A(x,y') = z$
and $A(x,y) = z' \ne z$.

Let $f: [n]^{k-1} \times [N] \to \{0,1\}$ be a graph function and let $A = Base(f)$. The communication complexity of $f$, $D_k(f)$,
is equal almost precisely to the minimal number of colors needed to color
the entries of $A$ so that no $A$-star is monochromatic \cite{CFL83, beigel2006multiparty, hdp17}.
It is also observed in \cite{BDPW07} (see also \cite{hdp17}) that $D_k(f)$, $N^1_k(f)$ and $D^1_k(f)$
are equivalent up to a small additive factor. Hence, any result on the deterministic communication complexity
of a graph function also holds (perhaps with slight change) for the nondeterministic and one-way complexity,
and vice versa.

As mentioned earlier previous known bounds are: (i) a lower bound of $\Omega(\log \log n)$ 
for the communication complexity of explicit graph functions
$f : [n]^{k-1} \times [N] \to \{0,1\}$ with $N \ll n$ \cite{BDPW07}, (ii) an 
$\Omega(\log \log \log n)$ lower bounds for the communication complexity
of any two-dimensional permutation \cite{alon2012nearly,graham2006monochromatic,beigel2006multiparty, hdp17}.
For higher-dimensional permutations the best lower bound is $\Omega(log^* n)$ \cite{hdp17},
but it is outside the scope of techniques discussed here.

All the above mentioned lower bounds, either for graph functions \cite{BDPW07} or for 
two-dimensional permutations 
\cite{alon2012nearly,graham2006monochromatic,beigel2006multiparty, hdp17} use the 
following general lower bound technique. For simplicity we sketch the technique for the case $k=3$. 
Let $f: [n]^2 \times [N] \to \{0,1\}$ be a graph function, and let $A=Base(f)$. Assume that $D_3(f) \le \log L$ for
some natural number $L$. This means that it is possible to color the entries of $A$ with
$L$ colors, so that no $A$-star is monochromatic. 
Following is an outline of a general lower bound technique for $L$: 

\newtheorem{alg}{}
\newcommand{\BA}{\begin{alg}} \newcommand{\EA}{\end{alg}}
\bigskip
\fbox{
\begin{minipage}{4.9in}
\begin{enumerate}

\item Let $E = [n]^2$ and $V = \emptyset$.

\item While $E$ contains entries whose value does not appear in $V$, do:

\begin{enumerate}

\item
Pick $v$, the most frequent value from $[N] \setminus V$ that appears in $E$.

\item
Pick $c \in [L]$, the most abundant color among $E$'s $v$-entries.

\item Let $S \subset E$ be the subset of entries with value $v$ and color $c$.
Clearly, $|S| \ge |E|/(N L)$.

\item Let $\bar{S}$ be the minimal combinatorial rectangle that contains $S$.

\item Set $E = \bar{S}$, $V = V \cup \{v\}$.

\end{enumerate}

\end{enumerate}

\end{minipage}
}
\bigskip

The heart of this lower bound technique is the fact that the entries of $\bar{S} \setminus S$ cannot be colored by the color $c$ or else there would be 
a monochromatic $A$-star. Thus $L$, the number of required colors, is at least the number of iterations of the above loop.
To prove a lower bound on the number of such iterations, it is necessary to prove a lower bound on the size
of the enclosing combinatorial rectangle $\bar{S}$. This bound on $\bar{S}$ determines the quality of the bound on $L$
achieved using the above technique.

The lower bounds of \cite{BDPW07} on the deterministic communication complexity of explicit graph functions,
and the lower bounds in \cite{hdp17} on the deterministic communication complexity of permutations,
and also related bounds \cite{beigel2006multiparty}, \cite{graham2006monochromatic} and \cite[Proposition~4.3]{alon2012nearly}, 
all follow the above scheme. In \cite{BDPW07} they use multicolor discrepancy \cite{BHK01} to
bound the size of $\bar{S}$, while in \cite{hdp17} and the related works, the structural properties of
a permutation are used to this end. 

The structural properties of permutations imply that $|\bar{S}| = |S|^2$, which 
gives the lower bound $\Omega(\log \log \log n)$. 
Discrepancy on the other hand gives better estimates on $\bar{S}$ and yields the bound 
$\Omega(\log \log n)$ on the communication complexity, which is perhaps the limit
of this general technique. Another strong advantage of
using discrepancy is that the bound works also for $k > 3$, and not only for $k=3$.
But the use of discrepancy seems limited to the case where $N \ll n$. Thus,
improving the known lower bounds on explicit graph functions as well as specifically the
much more limited bounds known for permutations, seems to require new ideas.

Note that even though the results of \cite{BDPW07} are similar to ours regarding
the lower bound that is achieved on the communication complexity of explicit graph functions,
the techniques are different. In fact the statements are in a way complementary as we now explain.
Let $f : [n]^{k-1} \times [N] \to \{0,1\}$ be a graph function whose base function has discrepancy
smaller than $O(\frac{1}{N^{1+\epsilon}})$, for some $\epsilon > 0$.
Let $\nh_{k-1, b}(Base(f)) = h+c$, where $h$ is the size of the 
help string and $c$ is the length of communication needed after the help string is given. 
In \cite{BDPW07} it is proved that as long as $h$ is much smaller than a constant times $\log\log n$ then 
$c$ is significantly larger than $\log N$. The bounds via nondeterministic communication complexity with help
on the other hand says that regardless of $h$, it might even be that $h = \log N -1$, 
$c$ is at least $\Omega(\log \log n)$.

\section{One round communication complexity and communication complexity with help}
\label{sec:other-lifts}

Let $A: [n]^{k-1} \to [N]$ be a function. We have defined $f=Lift(A)$ as the graph function associated with $A$,
and showed that $D_k(f)$ is strongly related to the nondeterministic communication complexity with help
of $A$. In a way, when we go from $A$ to $f$ we represent the value in each entry of $A$ by a boolean vector which
is the unary representation of this value. Denote this representation by $f = Lift_U(A)$, it is possible to consider other 
representations as well:
\begin{enumerate}

\item $f = Lift_B(A)$ is the function $f: [n]^{k-1} \times [N] \to \{0,1\}$ where 
$f(x_1, \ldots,x_{k-1},i)$ is equal to the $i$th bit in the binary representation of $A(x_1, \ldots,x_{k-1})$.

\item $f = Lift_{GT}(A)$ is the function $f: [n]^{k-1} \times [N] \to \{0,1\}$ satisfying 
$f(x_1, \ldots,x_k-1,y)=1$ iff $A(x_1,\ldots,x_{k-1})\ge y$.

\end{enumerate}

The representation $f = Lift_B(A)$ was considered in \cite{BHK01} in order to prove lower bounds
on the one-way communication complexity of $f$ in the NOF model, denoted $D^1_k(f)$. 
In fact, they mention this was one of the motivations for their paper. Similarly to the unary representation, 
it is not hard to check that $D^1_k(f) \le \dh_{k-1,b}(A) + 1$. It is proved in \cite{BHK01} that this relation goes both ways 
as long as $b$ is not too large.
\begin{lemma}[\cite{BHK01}]
Let $A: [n]^{k-1} \to [N]$ be a function, and let $f = Lift_B(A)$. Then
$$
D^1_{k}(f) \ge \min \{ \frac{1}{b}\dh_{k-1,b}(A), b\}.
$$
\end{lemma}

In the binary representation the last dimension of $f$ is small, $\log N$.
This was an advantage for \cite{BHK01} as one of their main applications was
to show that $D^1(f)$ can vary significantly when different players are allowed to speak first. 
But for the purpose of separating deterministic from randomized communication complexity
this is a disadvantage, since then $D_k(f)$ is bounded by $\log \log N$. 
A way to remedy this is to consider the representation $f = Lift_{GT}(A)$.

The representation $f = Lift_{GT}(A)$ is useful for our purposes since the dimensions
are not limited and also $R_k(f) \le \log \log N$ by a simple reduction to the two players "greater than" function. 
Similarly to the binary representation, the communication complexity
with help of $A$ is also closely related to one-way communication complexity of $f$.
It is again not hard to verify that $D^1_k(f) \le \dh_{k-1,\log N -1}(A)+1$, and on the other hand:
\begin{lemma}
\label{222}
Let $A: [n]^{k-1} \to [N]$ be a function, and let $f = Lift_{GT}(A)$. Then
$$
D^1_{k}(f) \ge \min \{ \frac{1}{\log N}\dh_{k-1,\log N -1}(A), \log N\}.
$$
\end{lemma}

\begin{proof}
If $D^1_k(f) \ge \log N$ then we are done. Otherwise, Alice and Bob perform a binary search
for the value of $A(x_1,\ldots,x_{k-1})$. They use as a help string
the transcript of the $k$-th player on input $(x_1,\ldots,x_{k-1},A(x_1,\ldots,x_{k-1}))$, 
in an optimal protocol for $D^1_k(f)$. 

Note that the transcript of the $k$-th player is independent of the $k$-th input,
and also independent of the other players transcript (as the protocol is one-way).
Thus the players can use this transcript to compute $f(x_1,\ldots,x_{k-1},y)$
for any $y \in [N]$. Each such computation would require at most $D^1_k(f)$
bits of computation. Using a binary search and at most $ D^1_k(f)\log N$ bits of communication, 
the players can compute this way the value of $A(x_1,\ldots,x_{k-1})$.
\end{proof}

A simple observation is that $D^1_k(Lift_U(A)) \le 2D^1_k(Lift_{GT}(A))$. Similarly, it also 
holds that $D_k(Lift_U(A)) \le 2D_k(Lift_{GT}(A))$ which means that the vast majority of 
functions $f=Lift_{GT}(A)$ are also good candidates for separating randomized
from deterministic communication complexity, as graph functions are. 

\section{Discussion and open problems}

Proving lower bounds on the deterministic communication complexity of
explicit graph functions $f : [n]^{k-1} \times [N] \to \{0,1\}$ is one of 
the most elementary open problems in the Number On The Forehead model.
Still, proving such a bound would most likely require new techniques that will help
with other problems in this area as well, and in fact also in other areas.
The deterministic NOF communication complexity
of permutations and linjections for example, which are a special family of graph functions,
have strong relations with well studied problems in other 
mathematical fields, and proving lower bound therein can have very interesting consequences 
such as lower bounds for the multidimensional Szemer\'{e}di theorem and 
{\em corners theorems}, lower bounds on the density of Ruzsa-Szemer\'{e}di 
graphs, a combinatorial proof for the Hales-Jewett theorem, and more. 
See e.g. \cite{CFL83, beigel2006multiparty, hdp17, shraibman2017note} for more details.

The best open problems are to prove stronger lower bound than $\Omega(\log \log n)$ on any explicit graph function,
improve the $\Omega(\log \log \log n)$ lower bound for a two-dimensional permutation, or the much weaker bounds
for higher-dimensional permutations. But there are also other related problems that are interesting, we list a few of them:

\begin{enumerate}

\item Determine the relation between nondeterministic communication complexity with help,
and multicolor discrepancy.

\item Determine the relation between nondeterministic communication complexity with help,
and deterministic communication complexity with help.

\item What is the maximal gap between $D^1_k(Lift_{GT}(A))$ and $D_k(Lift_{GT}(A))$ for a function $A: [n]^{k-1} \to [N]$? 
Any relation would enable to use Lemma~\ref{222} to lower bound $D_k(Lift_{GT}(A))$.

\item Find an explicit function  $A: [n]^{2} \to [N]$ with small discrepancy and large enough $N$, for which the gap between 
$D^1_k(Lift_{GT}(A))$ and $D_k(Lift_{GT}(A))$ is small.

\item What is the maximal gap between $D^1_k(Lift_U(A))$ and $D^1_k(Lift_{GT}(A))$ for a function $A: [n]^{k-1} \to [N]$?
Again, if there is a strong relation then
Lemma~\ref{222} gives a bound on $D_k(Lift_U(A))$ since for graph functions
one-way communication is as strong as regular protocols. 

\item Find an explicit function  $A: [n]^{2} \to [N]$ with small discrepancy and large enough $N$, for which the gap between $D^1_k(Lift_U(A))$ 
and $D^1_k(Lift_{GT}(A))$ is small.

\end{enumerate}

\bibliographystyle{plain}
\bibliography{../complexity}

\begin{thebibliography}{1}

\bibitem{alon2012nearly}
N.~Alon, A.~Moitra, and B.~Sudakov.
\newblock Nearly complete graphs decomposable into large induced matchings and
  their applications.
\newblock In {\em Proceedings of the forty-fourth annual ACM symposium on
  Theory of computing}, pages 1079--1090. ACM, 2012.

\bibitem{BHK01}
L.~Babai, T.~Hayes, and P.~Kimmel.
\newblock The cost of the missing bit: communication complexity with help.
\newblock {\em Combinatorica}, 21:455--488, 2001.

\bibitem{BDPW07}
P.~Beame, M.~David, T.~Pitassi, and P.~Woelfel.
\newblock Separating deterministic from randomized nof multiparty communication
  complexity.
\newblock In {\em Proceedings of the 34th International Colloquium On Automata,
  Languages and Programming}, Lecture Notes in Computer Science.
  Springer-Verlag, 2007.

\bibitem{beigel2006multiparty}
R.~Beigel, W.~Gasarch, and J.~Glenn.
\newblock The multiparty communication complexity of exact-t: Improved bounds
  and new problems.
\newblock In {\em International Symposium on Mathematical Foundations of
  Computer Science}, pages 146--156. Springer, 2006.

\bibitem{CFL83}
A.~Chandra, M.~Furst, and R.~Lipton.
\newblock Multi-party protocols.
\newblock In {\em Proceedings of the 15th ACM Symposium on the Theory of
  Computing}, pages 94--99. ACM, 1983.

\bibitem{graham2006monochromatic}
R.~Graham and J.~Solymosi.
\newblock Monochromatic equilateral right triangles on the integer grid.
\newblock In {\em Topics in discrete mathematics}, pages 129--132. Springer,
  2006.

\bibitem{KN97}
E.~Kushilevitz and N.~Nisan.
\newblock {\em Communication Complexity}.
\newblock Cambridge University Press, 1997.

\bibitem{hdp17}
N.~Linial and A.~Shraibman.
\newblock On the communication complexity of high-dimensional permutations.
\newblock {\em arXiv preprint arXiv:1706.02207}, 2017.

\bibitem{shraibman2017note}
A.~Shraibman.
\newblock A note on multiparty communication complexity and the hales-jewett
  theorem.
\newblock {\em arXiv preprint arXiv:1706.02277}, 2017.

\end{thebibliography}

\end{document}